\documentclass[letterpaper, 10 pt, conference]{ieeeconf}  

\IEEEoverridecommandlockouts                              
\usepackage{amsmath,amsfonts,amssymb, amsthm}
\usepackage{cuted}
\usepackage{xcolor}
\usepackage{graphicx}
\usepackage{parskip}
\usepackage{cite}
\usepackage{booktabs}
\usepackage{subcaption}

\usepackage{etoolbox}
\patchcmd{\thebibliography} 
  {\settowidth}
  {\setlength{\parsep}{0pt}\setlength{\itemsep}{0pt plus 0.1pt}\settowidth}
  {}{}

\newtheorem{lemma}{Lemma}
\newtheorem{theorem}{Theorem}

\newtheorem{assumption}{Assumption}
\newtheorem{remark}{Remark}

\renewcommand{\Re}{\mathbb{R}}

\renewcommand{\t}{^{\mbox{\tiny\sf T}}}

\renewcommand{\baselinestretch}{0.979}

\title{\LARGE \bf
Data-Driven Covariance Steering with Output Feedback
}

\author{Dimitrios Moustroufis$^{1}$ and Panagiotis Tsiotras$^{2}$
\thanks{$^{1}$ Ph.D. student, School of Aerospace Engineering, Georgia Institute of Technology, Atlanta, GA, 30332, USA. Email:
 {\small dimoustroufis@gatech.edu}}%
\thanks{$^{2}$ David and Andrew Lewis Chair and Professor, School of Aerospace Engineering, and Institute for Robotics and Intelligent Machines, Georgia Institute of Technology, Atlanta, GA, 30332, USA. Email:
 {\small tsiotras@gatech.edu}}%
}

\begin{document}

\maketitle
\thispagestyle{empty}
\pagestyle{empty}

\begin{abstract}
This paper addresses output-feedback covariance steering for stochastic discrete-time linear time-invariant systems with unknown dynamics.
We construct a controllable non-minimal state representation from past inputs and outputs, allowing the problem to be formulated in a standard state-feedback setting.
The induced disturbance, however, is temporally correlated, requiring propagation of the state-disturbance cross-covariance.
Using persistently exciting offline data, we employ an indirect approach for mean steering and a direct approach for covariance steering.
The indirect formulation requires estimation of the mean dynamics, whereas the direct formulation requires estimation of the noise realization.
We develop estimation methods suitable for temporally correlated noise and formulate the resulting covariance steering problem as a convex semidefinite program.
Numerical simulations demonstrate the effectiveness of the proposed framework.
\end{abstract}



\section{Introduction}
The increasing complexity of modern engineering systems, along with the wide availability of sensor data, have motivated recent research around data-driven control methods.
These methods either exploit data to obtain a model of the system dynamics to be used for model-based control design, or compose a control law directly from data.
These methods are classified as ``indirect" and ``direct," respectively.
Under the assumption of a linear, time-invariant system with no noise, the seminal work by Willems et al. \cite{willems_note_2005} proved that future trajectories of a system can be represented by a set of persistently exciting input-output data.
This result, often referred to as the ``Willems' Fundamental Lemma" (WFL) has been exploited for data-driven simulation \cite{markovsky_data-driven_2008}, for the design of feedback stabilizing controllers \cite{de_persis_formulas_2020, dai_data-driven_2023}, predictive controllers \cite{coulson_data-enabled_2019}, \cite{berberich_data-driven_2021}, as well as for data-enabled policy optimization algorithms \cite{zhao_data-enabled_2024}, \cite{zhao_direct_2025}.
These algorithms have been applied to control quadrotors, power networks, and synchronous motor drives \cite{elokda_dataenabled_2021, huang_decentralized_2022, carlet_data-driven_2022}.
Willems' result has been extended to data collected from multiple trajectories \cite{van_waarde_willems_2020}, while robustifying regularizations \cite{coulson_regularized_2019} have been employed to handle noise and nonlinearities in the data. 
Finally, the connection between indirect and direct methods based on the WFL has also been investigated in the literature~\cite{dorfler_bridging_2023}.

Covariance Steering (CS) is a new control design paradigm that explicitly accounts for uncertainty in the dynamics or the initial conditions by steering the entire distribution of the system state to some specified distribution, within a finite horizon.
For a known system model, the optimal solution of the CS problem, along with its formulation as a semi-definite program (SDP), has been established in the literature both for discrete-time \cite{bakolas_finite-horizon_2018, liu_optimal_2024, rapakoulias_discrete-time_2023, balci2022exact}, and for continuous-time systems \cite{chen2015optimal}.
 The CS problem for partially observed systems has also been recently addressed in the literature \cite{pilipovsky_computationally_2024}, \cite{maity_optimal_2025}.

Addressing the lack of a system model, the work in~\cite{pilipovsky_data-driven_2023} formulated the CS problem for deterministic systems in a data-driven fashion by exploiting the WFL, while \cite{pilipovsky_dust_2026} extended these results for systems with stochastic process noise.
In these works, an indirect approach is employed for the mean steering problem, while the covariance steering problem is formulated using a direct approach. 
A major standing assumption of the current data-driven covariance steering literature is the availability of full-state measurements without measurement noise. 
This, however, may not be feasible in many practical applications, thereby motivating the present work.

The main contribution of this paper is to extend the state-feedback, data-driven covariance steering results \cite{pilipovsky_data-driven_2023, pilipovsky_dust_2026}, to the output-feedback case.
By exploiting a non-minimal state representation \cite{de_persis_formulas_2020}, we utilize state-feedback ideas to solve the data-driven covariance steering problem. 
However, in this representation, the disturbances are temporally correlated, which complicates model estimation for the indirect mean controller design and noise estimation for the direct covariance controller design.
Hence, as a by-product, this work also addresses the discrete-time state-feedback covariance steering problem under temporally correlated noise, which, to the best of the authors' knowledge, has not been addressed in the literature. 

\section{Problem Formulation}

\subsection{Notation}

We use the operator $\mathrm{col}(a,b) \triangleq [a\t\ b\t]\t$
To denote a stacked vector composed of the column vectors $a,b$. 
Considering the $(T+1)$-long vector sequence $(a_0, a_1, \ldots, a_{T})$, we denote its Hankel matrix of depth $L$ and length $T$ with $\mathcal{H}_{i,L,T}(a)$. 
The $j$-th column of $\mathcal{H}_{i,L,T}(a)$ is $\mathrm{col}(a_{i+j}, \ldots, a_{i+j+L-1})$, $j \in [0,\ T-1]$. 
We also denote with $\mu^{(\cdot)}_k$ and $\Sigma^{(\cdot)}_k$ the mean and covariance of the random variable in the superscript at some time $k$. 
For a symmetric positive (semi-)definite matrix $A$ we use the notation $A \succ 0$ ($\succeq 0$), and we denote the Moore-Penrose pseudo-inverse of $A$, as $A^\dagger$.

\subsection{Problem Statement}

Consider the stochastic, discrete-time LTI system
\begin{equation}
    \begin{aligned}\label{eq:x_dynamics}
        x_{k+1} &= Ax_k + Bu_k + w_k,\\
        y_k &= Cx_k + q_k,
    \end{aligned}
\end{equation}
where $x_k,w_k\in \mathbb{R}^n$, $y_k, q_k\in\mathbb{R}^p$, $u_k\in \mathbb{R}^m$ and $w_k,q_k$ are identically distributed samples from Gaussian noise processes that represent the process and measurement noise, respectively.
While $w,q$ can be correlated in time (i.e., ``colored"), they must be stationary and mutually independent.
The data-driven output distribution steering problem is posed as steering the output from some initial distribution $\mathcal{N}(\mu_i^y, \Sigma_i^y)$ to a final distribution $\mathcal{N}(\mu_f^y, \Sigma_f^y)$ over a specified horizon $N$, while minimizing some cost index, without knowledge of the model matrices $A,B,C$, but rather utilizing only $T+1$ input-output data pairs from a single trajectory.
%
%
%

\subsection{Non-minimal State Space Representation}

We define a non-minimal state space representation as
\begin{align}
    z_k \triangleq \mathrm{col}(u_{k-1}, \ldots, u_{k-\ell}, y_{k-1}, \ldots, y_{k-\ell}) \in \mathbb{R}^h,
\end{align}
where $\ell$ denotes the observability index of the system, i.e, the smallest integer such that the observability matrix $\mathcal{O} = [(CA^{\ell-1})\t\ \cdots\ (CA)\t\quad C\t]\t$ has rank $n$, and $h=\ell(m+p)$. 
In case $\ell$ is not known, an upper bound of $\ell$ can be used.
The system order $n$ serves as the largest such upper bound, and it can be estimated from input-output measurements \cite{de_persis_formulas_2020}.
Define the vectors $\tilde{q}_k \triangleq \mathrm{col}(q_{k-1}, \ldots, q_{k-\ell})$ and $\tilde{w}_k \triangleq \mathrm{col}(w_{k-1}, \ldots, w_{k-\ell})$.
The dynamics of the non-minimal state $z_k$ can then be expressed as 
\begin{align} \label{eq:z_dynamics}
  z_{k+1}  &= A_zz_k + B_zu_k + \tilde{\xi}_k,\\
    y_k &= \mathcal{S}z_k + \tilde{\zeta}_k, \label{eq:z_dynamicsB}
\end{align}
where $\tilde{\xi}_k = D_z\tilde{\zeta}_k$, and $\tilde{\zeta}_k$ is given by
\begin{equation}\label{eq:d_expression}
    \begin{aligned}
        \tilde{\zeta}_k &\triangleq C(\mathcal{C}_w - A^{\ell}\mathcal{O}^\dagger\mathcal{T}_{w})\tilde{w}_k - CA^{\ell}\mathcal{O}^\dagger\tilde{q}_k + q_k \\
            &= F_w \tilde{w}_k + F_q\tilde{q}_k + q_k.
    \end{aligned}
\end{equation}
and where $A_z, B_z, D_z, \mathcal{S}, \mathcal{C}_w, \mathcal{T}_w$ are provided in the Appendix.

\begin{lemma}\label{lem:d_noise_correlation}
 The samples $\tilde{\zeta}_k$
 are zero-mean, Gaussian, identically distributed, but potentially correlated, even for i.i.d. processes $w, q$.
 In particular, assuming that $w_k, q_k$ are i.i.d., for the time points $k,m$, the vectors $\tilde{\zeta}_k$ and $\tilde{\zeta}_m$ can be correlated if $|m-k| \le \ell$, and they are independent otherwise.
\end{lemma}

\textit{Proof.} See the Appendix. 
\qed

It has been shown in \cite{alsalti_notes_2025} that, for $p\ell>n$ and for systems without noise, the system \eqref{eq:z_dynamics} is not controllable. 
Moreover, the matrix $\begin{bmatrix} \mathcal{H}_{0,1,T}(u)\t \mathcal{H}_{0,1,T}(z)\t \end{bmatrix}\t$ does not have full row rank, thus violating a fundamental condition of the data-driven control methods used in the sequel. 
Following~\cite{alsalti_notes_2025}, we employ the reduced non-minimal state
\begin{equation}\label{eq:chi_definition}
\chi_k \triangleq L z_k \in \mathbb{R}^r,
\end{equation}
where $L \in \Re^{r \times h}$ removes the linearly dependent components of $z_k$. 
Details on computing $L$ from noisy data are presented in Section~\ref{sec:data_driven_system_representation}. 
The dynamics of $\chi_k$ are then written as
\begin{subequations} \label{eq:chi_dynamics}
    \begin{align}
        \chi_{k+1} &= \mathcal{A}\chi_k + \mathcal{B}u_k + \mathcal{D}\zeta_k = \mathcal{A}\chi_k + \mathcal{B}u_k + \xi_k, \\
        y_k &= \mathcal{C} \chi_k + \zeta_k,
    \end{align}
\end{subequations}
where $\xi_k \triangleq \mathcal{D}\zeta_k$, and $\mathcal{A},\mathcal{B}, \mathcal{C}, \mathcal{D}$ will be estimated from data.
The noise $\zeta_k$ is temporally correlated, as is $\tilde{\zeta}_k$.

\subsection{Non-minimal State Distribution Steering Problem}\label{sec:nm_distribution_steering_problem}

In the absence of probabilistic chance constraints, we can separate the distribution steering problem for the state $\chi_k$ into the mean and covariance steering problems \cite{okamoto_optimal_2018}.
Moreover, we consider a controller of the form 
\begin{equation}  \label{eq:control}
u_k = K_k(\chi_k -\mu_k) + v_k + \nu_k,
\end{equation}
where $v_k$ is a feed-forward term, $K_k$ is a feedback gain matrix, and $\nu_k~\sim \mathcal{N}(0, \Sigma^\nu_k)$ is an independent random term.
As shown in \cite{liu_optimal_2024} this class of controllers includes the optimal one for distribution steering problems.
Defining $\mu_k\triangleq \mathbb{E}[\chi_k]$ and $\Sigma_k \triangleq \mathbb{E}[(\chi_k-\mu_k)(\chi_k-\mu_k)\t]$, as well as $\mu_k^y\triangleq \mathbb{E}[y_k]$ and $\Sigma_k^y \triangleq \mathbb{E}[(y_k-\mu_k^y)(y_k-\mu_k^y)\t]$, the mean output steering problem is posed as
\begin{subequations}\label{eq:mean_problem_chi_state}
\begin{equation}
    \begin{aligned}
        \min_{\mu_k, v_k} J = \sum_{k=0}^{N-1}(\mu_k - \chi_k^r)\t Q^\mu_k(\mu_k - \chi^r_k) + v_k\t R^\mu_k v_k \label{eq:mean_steering_cost},
    \end{aligned}
\end{equation}
such that, for all $k=0,1,\dots, N-1$,
\begin{align}
    &\mu_{k+1} = \mathcal{A}\mu_k + \mathcal{B}v_k, \label{eq:mean_dynamics_constraint}\\
    &\mu_0^y = \mu_i^y, \quad
    \mu_N^y = \mu_f^y \label{eq:init_term_mean_constraints},
\end{align}
\end{subequations}
where $v_k$ is a feedforward control input, $\chi_k^r$ is a reference value for $\chi_k$, and $Q^\mu_k \succeq 0$, $R^\mu_k \succ 0$ are weight matrices.

From the dynamics~\eqref{eq:chi_dynamics}, we notice that the state $\chi_k$ is correlated with the noise $\zeta_{k-1}$. 
However, from Lemma~\ref{lem:d_noise_correlation}, $\zeta_{k-1}$ is correlated with $\zeta_k$, thus making $\chi_k$ correlated with $\zeta_k$ as well. This means that the cross-covariance between $\zeta_k$ and $\chi_k$ will not be zero, and this has to be accounted for.

In this paper, we model $\zeta_k$ as the stationary AR(1) process 
\begin{align}
    \zeta_{k+1} &= \Psi \zeta_k + \eta_k, \label{eq:zeta_noise_dynamics}
\end{align}
where the spectral radius, $\rho(\Psi)$, satisfies $\rho(\Psi)<1$, and $\eta_k$ is a zero-mean, i.i.d. Gaussian process.
Since the process is stationary, the covariance of $\zeta$ is constant, that is $\Sigma^\zeta_k \equiv \Sigma^\zeta$.
The AR(1) assumption is commonly used to model colored noise in practice and, in our case, it enables tractable covariance propagation.
Using the control law~\eqref{eq:control}, the propagation equation for $\Sigma_k$  is written as
\begin{equation}\label{eq:chi_covar_propagation}
\begin{aligned}
    \Sigma_{k+1} &= \begin{bmatrix} \mathcal{B} & \mathcal{A} \end{bmatrix} \begin{bmatrix} K_k \\ I_r \end{bmatrix} \Sigma_k \begin{bmatrix} K_k \\ I_r \end{bmatrix}\t \begin{bmatrix} \mathcal{B} & \mathcal{A} \end{bmatrix}\t\\
    &+ \begin{bmatrix} \mathcal{B} & \mathcal{A} \end{bmatrix} \begin{bmatrix} K_k \\ I_r \end{bmatrix} \Sigma^{\chi\zeta}_k \mathcal{D}\t +
    (\Sigma^{\chi\zeta}_k \mathcal{D}\t)\t \begin{bmatrix} K_k \\ I_r \end{bmatrix}\t \begin{bmatrix} \mathcal{B} & \mathcal{A} \end{bmatrix}\t\\
    &+\mathcal{B}\Sigma^\nu_k\mathcal{B}\t + \mathcal{D} \Sigma^{\zeta} \mathcal{D}\t,
\end{aligned}
\end{equation}
where $\Sigma^{\chi\zeta}_{k} \triangleq \mathbb{E}[\chi_{k} \zeta_{k}\t]$ is propagated as
\begin{align}\label{eq:chi_xi_covar_propagation}
    \Sigma^{\chi\zeta}_{k+1} = \begin{bmatrix} \mathcal{B} & \mathcal{A} \end{bmatrix} \begin{bmatrix} K_k \\ I_r \end{bmatrix} \Sigma^{\chi\zeta}_k \Psi\t + \mathcal{D} \Sigma^\zeta \Psi\t.
\end{align}
The full derivations are provided in the Appendix.

Given the weight matrices $Q^\Sigma_k \succeq 0$ and $R^\Sigma_k\succ0$, the output covariance steering problem is posed as
\begin{equation}\label{eq:covar_problem_chi_state}
\begin{aligned}
    \min_{\Sigma_k, \Sigma^{\chi\zeta}_k, K_k, \Sigma^\nu_k} J = \sum_{k=0}^{N-1}&\mathrm{tr}(Q^\Sigma_k\Sigma_k) + \mathrm{tr}(R^\Sigma_k(K_k\Sigma_kK_k\t + \Sigma^\nu_k)),
\end{aligned}
\end{equation}
such that the covariance propagation equations \eqref{eq:chi_covar_propagation}, \eqref{eq:chi_xi_covar_propagation} are satisfied for all $k=0,1,\dots, N-1$, and the boundary conditions $\Sigma_0^y = \Sigma_i^y$, $\Sigma_N^y = \Sigma_f^y$, $\Sigma^{\chi\zeta}_0 = \Sigma^{\chi\zeta}_i$ are met.

%

\section{Data-Driven System Representation}\label{sec:data_driven_system_representation}

We formulate the mean steering problem using an indirect method and the covariance steering problem using a direct method.
For the former problem, an estimated model of the mean dynamics is required, while the latter problem requires an estimate of the noise realization in the data \cite{pilipovsky_data_driven_2024}.
To address these requirements, we present two methods to estimate the matrices $\mathcal{A},\mathcal{B}$ in \eqref{eq:chi_dynamics}.

\subsection{Persistency of Excitation and Multi-Output Systems}

Consider the data matrices constructed from a $(T+1)$-long trajectory, namely, $U_{0,T} \triangleq \mathcal{H}_{0,1,T}(u)$, $Z_{0,T} \triangleq \mathcal{H}_{0,1,T}(z)$, $Z_{1,T} \triangleq \mathcal{H}_{1,1,T}(z)$, $\mathcal{X}_{0,T} = LZ_{0,T}$, $\mathcal{X}_{1,T}=LZ_{1,T}$, $Y_{0,T} \triangleq \mathcal{H}_{0,1,T}(y)$, $\Xi_{0,T} \triangleq \mathcal{H}_{0,1,T}(\xi)$, $\mathcal{Z}_{0,T} \triangleq \mathcal{H}_{0,1,T}(\zeta)$.
    
\begin{theorem}\label{lem:WFL}
    (Willems' Fundamental Lemma \cite{willems_note_2005}, expressed for single-step predictions). 
    Letting $\zeta_k \equiv 0$ in the discrete-time LTI system \eqref{eq:z_dynamics}, and assuming that the system is controllable and that the sequence $U_{0,T}$ is persistently exciting of order $h+1$, then the matrix $S \triangleq \begin{bmatrix} U_{0, T}\t & Z_{0,T}\t \end{bmatrix}\t$ has full row rank, that is, $\mathrm{rank}(S) = m + h$.
\end{theorem}

For the common case  $p\ell >n$, $S$ is rank deficient \cite{alsalti_notes_2025}.
To address this problem, following \cite{alsalti_notes_2025}, we have introduced the non-minimal state space representation $\chi_k$ in \eqref{eq:chi_definition}, where the matrix $L$ is such that $\begin{bmatrix} U_{0,T} \\ LZ_{0,T} \end{bmatrix}$ has full row rank equal to $m(\ell +1) +n$.
To compute $L$ from noisy data, we compute the singular value decomposition (SVD) of $S$ as
\begin{align}\label{eq:z_svd}
S = \begin{bmatrix}
    U_1 & U_2
\end{bmatrix} \begin{bmatrix}
    \Lambda_1 & 0 \\
    0 & \Lambda_2
\end{bmatrix} \begin{bmatrix}
    V_1\t \\ V_2\t
\end{bmatrix},
\end{align}
where the partitioning of the block matrices is such that $\Lambda_1$ contains the $r$ largest singular values. Ideally, $r=m\ell+n$, but if $n$ is unknown, we select the $r > m\ell$ singular values that are clearly distinguishable from the rest. 
The matrix $L$ is computed as $L=\Lambda_1^{-1}U_1\t$.
%
%
Finally, we define the data matrix
\begin{align}\label{eq:data_matrix_P}
    P \triangleq \begin{bmatrix}
        U_{0,T} \\
        \mathcal{X}_{0,T}
    \end{bmatrix} = \begin{bmatrix}
        U_{0,T} \\
        LZ_{0,T}
    \end{bmatrix},\ P\in \Re^p.
\end{align}
It can be easily shown that the collected data must satisfy the following equations
\begin{align}
    \mathcal{X}_{1,T} &= \begin{bmatrix}
        \mathcal{B} & \mathcal{A}
    \end{bmatrix}P + \Xi_{0,T} \label{eq:chi_data_dynamics},\\
    Y_{0,T} &= \mathcal{C} \mathcal{X}_{0,T} + \mathcal{Z}_{0,T} \label{eq:output_data_dynamics}.
\end{align}

\subsection{Total Least Squares Method}

The standard approach to solve the model estimation problem would be to exploit the ordinary least squares (LS) method with $P$ in \eqref{eq:chi_data_dynamics} as the regressors and $\mathcal{X}_{1,T}$ as the observations.
However, for our problem, uncertainty exists both in the regressor matrix and in the observations, thus making the Total Least Squares (TLS) \cite{van_huffel_total_2004} approach more suitable than standard LS. 
The TLS problem is formulated as
\begin{subequations}
    \begin{align*}
    \min_{\mathcal{A}, \mathcal{B}, \Delta\mathcal{X}, \Delta P} \left \Vert \begin{bmatrix} \Delta\mathcal{X} & \Delta P \end{bmatrix}\right \Vert_\mathrm{F},
    \end{align*}
    such that $ \mathcal{X}_{1,T} + \Delta \mathcal{X} = \begin{bmatrix}\mathcal{B} & \mathcal{A} \end{bmatrix}(P + \Delta P)$.
\end{subequations}

This problem admits a closed-form solution via the SVD of the augmented matrix $U\Lambda V\t = \begin{bmatrix}P\t & \mathcal{X}_{1,T}\t \end{bmatrix}$. 
Partitioning the matrices $U,V$ as $U=[U_1\ U_2]$, $V=[V_1\ V_2]$, where $U_1, V_1$ contain the first $p$ columns of $U, V$, respectively, and further partitioning $V$ as $\begin{bmatrix} V_{11} & V_{12} \\ V_{21} & V_{22} \end{bmatrix}$, where $V_{11}\in \Re^{p \times p}$, the estimated model is computed as
\begin{align}
    \begin{bmatrix} \hat{\mathcal{B}} & \hat{\mathcal{A}} \end{bmatrix} = (-V_{12} V_{22}^{-1})\t.
\end{align}
Moreover, the values of $\Delta \mathcal{X}, \Delta P$ are computed as
\begin{align}
    \begin{bmatrix}
        \Delta P\t & \Delta \mathcal{X}\t
    \end{bmatrix} = -U_2\Lambda_2 V_2\t.
\end{align}

\subsection{Instrumental Variables Method}

As discussed in Section~\ref{sec:nm_distribution_steering_problem}, $\chi_k$ is correlated with $\zeta_k$. In such cases, where the regressors are correlated with the noise, the TLS estimator will not be consistent, necessitating another approach. 
Here, we apply the Instrumental Variables (IV) method~\cite{wooldridge_econometric_2010}.

\begin{remark}
    Despite the IV estimator being theoretically consistent, it is less data-efficient than the TLS estimator, and it performs poorly both in terms of mean error and error variance when the available sample size is small. Thus, TLS serves as an alternative method for such cases.
\end{remark}

In the IV method, the ``instruments" are a set of variables, independent from $\zeta_k$, that are used as the independent variables of the problem.
An attractive choice of instruments in our case is the control input at time $k$ and the $b$ past control inputs, i.e., $g_k \triangleq [u_{k}\t\ u_{k-1}\t\ \cdots\ u_{k-b}\t]$.
To present the IV method, we consider the parameters $\beta \triangleq \begin{bmatrix} \mathcal{B} & \mathcal{A} \end{bmatrix}$, the regressors $j_k \triangleq [u_k\t\ \chi_k\t]$, the regressor matrix $J \triangleq P\t$, and the observations $Y \triangleq \mathcal{X}_{1,T}\t$. 
%
%
We consider the partition $j_k = \begin{bmatrix} j_{k,1}\ j_{k,2} \end{bmatrix}$, where $j_{k,1}$ are the exogenous components of $j_k$ (i.e., those not correlated with the noise), and $j_{k,2}$ are the endogenous components of $j_k$. 
The goal is to express $j_{k,2}$ as a linear combination of $g_k$ and $j_{k,1}$. To this end, we employ the two-stage least squares (2SLS) estimator, which chooses the linear combinations of $g_k$ with the highest correlation with $j_{k,2}$ \cite{wooldridge_econometric_2010} and is given by
\begin{align}\label{eq:2sls_formula}
    \hat{\beta} = \left [(J\t G(G\t G)^{-1}G\t J)^{-1}(J\t G(G\t G)^{-1}G\t Y) \right]\t,
\end{align}
where $G\triangleq [g_0\ g_1\ \cdots\ g_{T-1}]$ is the matrix of instrument vectors over all samples.
In order for the 2SLS estimator to be consistent, the following assumptions must hold:
\begin{assumption}\label{assum:iv_assum_1}
    The instruments are independent of the noise, $\mathbb{E}(g_k\zeta_k\t) =0$.
\end{assumption}

\begin{assumption}\label{assum:iv_assum_2}
    The instruments are linearly independent, $\mathrm{rank}(\mathbb{E}(g_kg_k\t)) = b + 1$.
\end{assumption}

\begin{assumption}\label{assum:iv_assum_3}
    The instruments are linearly dependent with the regressors to some degree, $\mathrm{rank}(\mathbb{E}(g_kj_k\t)) = r+m$.
\end{assumption}

If the control inputs used for data collection are mutually independent and independent of the noise, then Assumptions~\ref{assum:iv_assum_1},~\ref{assum:iv_assum_2} are satisfied.
A necessary condition for Assumption~\ref{assum:iv_assum_3} to be satisfied is that the number of instruments to be greater than or equal to the number of regressors, but, in practice, this assumption is verified empirically by computing the expectation from sample data.
The following theorem establishes the consistency of the 2SLS estimator under the aforementioned assumptions.
 
\begin{theorem}
    (Theorem 5.1 \cite{wooldridge_econometric_2010}) 
    If Assumptions~\ref{assum:iv_assum_1}, \ref{assum:iv_assum_2}, \ref{assum:iv_assum_3} are satisfied, then the 2SLS estimator is consistent.
\end{theorem}

\section{Data-Driven Distribution Steering}


\subsection{Data-driven Mean Steering}

The mean steering problem \eqref{eq:mean_problem_chi_state} is trivially expressed in terms of an estimated model, since the cost \eqref{eq:mean_steering_cost} remains the same, while the dynamics constraints \eqref{eq:mean_dynamics_constraint} are written in terms of $\hat{\mathcal{A}}, \hat{\mathcal{B}}$ instead of $\mathcal{A},\mathcal{B}$.
For the initial and terminal output mean constraints \eqref{eq:init_term_mean_constraints}, we need to estimate the matrix $\mathcal{C}$ in \eqref{eq:chi_dynamics}. 
An estimate, $\hat{\mathcal{C}}$, is obtained using TLS / IV with $\mathcal{X}_{0,T}$ as the regressors and $Y_{0,T}$ as the observations. Considering the instrument matrix $G$ as in \eqref{eq:2sls_formula}, and defining the regressor and observation matrices $J_c \triangleq \mathcal{X}_{0,T}\t$ and $Y_c \triangleq Y_{0,T}\t$, respectively, the IV estimator for $\mathcal{C}$ is given by
\begin{align}\label{eq:2sls_formula_c_matrix}
    \hat{\mathcal{C}} = \left [(J_c\t G(G\t G)^{-1}G\t J_c)^{-1}(J_c\t G(G\t G)^{-1}G\t Y_c) \right]\t.
\end{align}
The initial and terminal mean constraints are now posed as $\mu^y_i = \hat{\mathcal{C}}\mu_0$ and $\mu^y_f = \hat{\mathcal{C}}\mu_N$.

\subsection{Data-driven Covariance Steering}

Using the data matrix $P$ from \eqref{eq:data_matrix_P}, we construct the matrix $\Phi = PP\t$.
For persistently exciting data, Theorem~\ref{lem:WFL} implies that $P$ will have full row rank, and thus, $\Phi$ will have full rank.
Hence, there exists a solution $G_k\in \mathbb{R}^{(m+r) \times r}$ to the following system of linear equations
\begin{align}\label{eq:rouche_capelli}
    \begin{bmatrix}
        K_k & I
    \end{bmatrix}\t =\Phi G_k.
\end{align}

From \eqref{eq:rouche_capelli}, we can express the feedback gains $K_k$ in \eqref{eq:chi_covar_propagation} and \eqref{eq:chi_xi_covar_propagation} in terms of the data matrix $\Phi$ and the decision variables $G_k$. Using \eqref{eq:chi_data_dynamics}, we derive the following model-independent form of the covariance propagation equations
\begin{align}
    \Sigma_{k+1} &= (\mathcal{X}_{1,T} - \mathcal{D}\mathcal{Z}_{0,T}) P\t G_k \Sigma_k G_k\t P(\mathcal{X}_{1,T} - \mathcal{D}\mathcal{Z}_{0,T})\t \notag \\
    &+ (\mathcal{X}_{1,T}-\mathcal{D}\mathcal{Z}_{0,T}) P\t G_k \Sigma^{\chi\zeta}_k \mathcal{D}\t  \notag \\
    &+ \mathcal{D}(\Sigma^{\chi\zeta}_k)\t G_k\t P (\mathcal{X}_{1,T}-\mathcal{D}\mathcal{Z}_{0,T})\t + \mathcal{D} \Sigma^\zeta \mathcal{D}\t, \label{eq:dd_chi_covar_propagation}\\
    \Sigma^{\chi\zeta}_{k+1} &= (\mathcal{X}_{1,T}-\mathcal{D}\mathcal{Z}_{0,T}) P\t G_k \Sigma^{\chi\zeta}_k \Psi\t + \mathcal{D}\Sigma^\zeta \Psi\t. \label{eq:dd_chi_xi_covar_propagation}
\end{align}
In this data-driven formulation, the matrices $\mathcal{A}, \mathcal{B}$ are not parameterized separately, but the whole feedback interconnection is expressed in terms of the data. Thus, we cannot express the term $\mathcal{B}\Sigma^\nu_k \mathcal{B}\t$ in \eqref{eq:chi_covar_propagation} in a model-independent fashion, which is the reason for using the sub-optimal control law $u_k = K_k(\chi_k-\mu_k) + v_k$ in \eqref{eq:dd_chi_covar_propagation} instead of $u_k = K_k(\chi_k-\mu_k) + v_k + \nu_k$.

Next, the output covariance is expressed in terms of the optimization variables as
\begin{align}
    \Sigma^y_k &= \mathbb{E}[(y_k-\mu^y_k)(y_k-\mu^y_k)\t] \notag\\
               &= \mathcal{C}\Sigma_k \mathcal{C}\t + \mathcal{C} \Sigma^{\chi \zeta}_k + (\Sigma^{\chi \zeta}_k)\t\mathcal{C}\t + \Sigma^\zeta. \label{eq:y_covar_model}
\end{align}
See the Appendix for the full derivation. Since $\mathcal{C}$ is unknown, we will use $\hat{\mathcal{C}}$ from~\eqref{eq:2sls_formula_c_matrix} to implement \eqref{eq:y_covar_model}.

Furthermore, from the first block equation of \eqref{eq:rouche_capelli}, the objective function of the optimization is expressed as
\begin{align}\label{eq:dd_covar_problem_chi_state_cost}
    J = \sum_{k=0}^{N-1} &\mathrm{tr}(Q^\Sigma_k\Sigma_k)+ \mathrm{tr}(R^\Sigma_k \Theta G_k \Sigma_k G_k\t \Theta\t),
\end{align}
where $\Theta = \begin{bmatrix} U_{0,T}U_{0,T}\t & U_{0,T}\mathcal{X}_{0,T}\t \end{bmatrix}$.

To implement equations \eqref{eq:dd_chi_covar_propagation}-\eqref{eq:y_covar_model}, the noise realization $\mathcal{Z}_{0,T}$ and noise covariance $\Sigma^\zeta$, the dynamics and noise model matrices $\mathcal{D}, \Psi$, and the initial cross-covariance $\Sigma^{\chi\zeta}_i$ must be estimated.
To this end, we first approximate the realizations of $\xi$ and $\zeta$ from \eqref{eq:chi_data_dynamics} and \eqref{eq:output_data_dynamics} as the residuals between the observations and the predictions of an estimated model.
Then, the covariances of $\xi_k$ and $\zeta_k$ are computed as the sample covariances of the estimated realizations, as follows
\begin{align}
    \hat{\Xi}_{0,T} &= \mathcal{X}_{1,T} - \begin{bmatrix} \hat{\mathcal{B}} & \hat{\mathcal{A}} \end{bmatrix}P,\quad \quad
    \hat{\Sigma}^\xi = \frac{1}{T} \hat{\Xi}_{0,T} \hat{\Xi}_{0,T}\t\\
    \hat{\mathcal{Z}}_{0,T} &= Y_{0,T} - \hat{\mathcal{C}}\mathcal{X}_{0,T},\quad \quad \hat{\Sigma}^\zeta = \frac{1}{T} \hat{\mathcal{Z}}_{0,T} \hat{\mathcal{Z}}_{0,T}\t.
\end{align}
Since the noise process is stationary, it follows that $\hat{\Sigma}^\xi_k \equiv \hat{\Sigma}^\xi$ and $\hat{\Sigma}^\zeta_k \equiv \hat{\Sigma}^\zeta$. 
Additionally,  an estimate $\hat{\Psi}$ of the noise model \eqref{eq:zeta_noise_dynamics} is obtained using TLS with the estimated realizations $\hat{\mathcal{Z}}_{0,T-1} \triangleq \mathcal{H}_{0,1,T-1}(\hat{\zeta})$ as the regressors and $\hat{\mathcal{Z}}_{1,T-1} \triangleq \mathcal{H}_{1,1,T-1}(\hat{\zeta})$ as the observations. Furthermore, using ordinary least squares on the relation $\xi_k=\mathcal{D}\zeta_k$, we compute an approximation, $\hat{\mathcal{D}}$, of $\mathcal{D}$ using $\hat{\mathcal{Z}}_{0,T}$ as the regressors and $\hat{\Xi}_{0,T}$ as the observations. Finally, expressing the output equation of \eqref{eq:chi_dynamics} for the estimated model, multiplying from the right by $\chi_k\t$, and taking expectations, an estimate, $\hat{\Sigma}^{\chi\zeta}_i$, of the initial cross-covariance is obtained as
\begin{align}\label{eq:cross_covar_estimation}
    \hat{\Sigma}^{\chi\zeta}_i = (\mathbb{E}[y_{k}\chi_k\t] - \hat{\mathcal{C}}\mathbb{E}[\chi_k\chi_k\t])\t,
\end{align}
where the expectations are computed as $\mathbb{E}[y_k\chi_k\t] =  Y_{0,T} \mathcal{X}_{0,T}\t/T$ and $\mathbb{E}[\chi_k \chi_k\t] =\mathcal{X}_{0,T} \mathcal{X}_{0,T}\t/T$.

Equations \eqref{eq:dd_chi_covar_propagation}, \eqref{eq:dd_chi_xi_covar_propagation}, and \eqref{eq:dd_covar_problem_chi_state_cost}, yield a non-convex optimization problem due to the terms $G_k\Sigma^{\chi\zeta}_k$ and $G_k \Sigma_k G_k\t$. To alleviate this issue, we define the new decision variables $U_k \triangleq G_k\Sigma_k$ and $W_k \triangleq G_k\Sigma^{\chi\zeta}_k$. The term $G_k \Sigma_k G_k\t$ now becomes $U_k \Sigma^{-1}_k U_k\t$, which is still non-convex. In order to convexify the problem, we introduce the slack variable $Y_k \succeq U_k \Sigma_k^{-1} U_k\t$.

The relaxed optimization problem is finally formulated as
\begin{subequations}\label{eq:convex_dd_covar_problem_chi_state}
\begin{equation}
\begin{aligned}
    \min_{\Sigma_k, \Sigma^{\chi\zeta}_k, U_k, W_k, Y_k} \sum_{k=0}^{N-1} &\mathrm{tr}(Q^\Sigma_k\Sigma_k)+ \mathrm{tr}(R^\Sigma_k \Theta Y_k \Theta\t),
\end{aligned}
\end{equation}
such that, for all $k=0,1,\dots, N-1$,
\begin{align}
    \Sigma_{k+1} &= (\mathcal{X}_{1,T} - \hat{\mathcal{D}}\hat{\mathcal{Z}}_{0,T}) P\t Y_k P(\mathcal{X}_{1,T} - \hat{\mathcal{D}}\hat{\mathcal{Z}}_{0,T})\t \notag \\
    &+ (\mathcal{X}_{1,T}-\hat{\mathcal{D}}\hat{\mathcal{Z}}_{0,T}) P\t W_k \hat{\mathcal{D}}\t  \notag \\
    &+ \hat{\mathcal{D}} W_k\t P (\mathcal{X}_{1,T} - \hat{\mathcal{D}} \hat{\mathcal{Z}}_{0,T})\t + \hat{\mathcal{D}} \hat{\Sigma}^\zeta \hat{\mathcal{D}}\t, \allowdisplaybreaks\\
    \Sigma^{\chi\zeta}_{k+1} &= (\mathcal{X}_{1,T}- \hat{\mathcal{D}} \hat{\mathcal{Z}}_{0,T}) P\t W_k \hat{\Psi}\t + \hat{\mathcal{D}} \hat{\Sigma}^\zeta \hat{\Psi}\t, \allowdisplaybreaks\\
    &\begin{bmatrix}
        \Sigma_k & U_k\t & \Sigma^{\chi\zeta}_k \hat{\mathcal{D}}\t\\
        U_k & Y_k & W_k \hat{\mathcal{D}}\t \\
        \hat{\mathcal{D}}(\Sigma^{\chi\zeta}_k)\t & \hat{\mathcal{D}} W_k\t & \hat{\mathcal{D}} \hat{\Sigma}^\zeta \hat{\mathcal{D}}\t
    \end{bmatrix} \succeq 0 \label{eq:joint_covar_constraint},\allowdisplaybreaks\\
    \Sigma_k &= \begin{bmatrix} \mathcal{X}_{0,T} U_{0,T}\t &  \mathcal{X}_{0,T}\mathcal{X}_{0,T}\t\end{bmatrix} U_k \label{eq:rouche_capelli_second_block_chi_covar},\allowdisplaybreaks\\
    \Sigma^{\chi \zeta}_k &= \begin{bmatrix} \mathcal{X}_{0,T} U_{0,T}\t &  \mathcal{X}_{0,T}\mathcal{X}_{0,T}\t\end{bmatrix} W_k \label{eq:rouche_capelli_second_block_chi_xi_covar},\allowdisplaybreaks\\
    \Sigma^y_f &= \hat{\mathcal{C}}\Sigma_N \hat{\mathcal{C}}\t + \hat{\mathcal{C}} \Sigma^{\chi \zeta}_N + (\Sigma^{\chi \zeta}_N)\t\hat{\mathcal{C}}\t + \hat{\Sigma}^\zeta,\allowdisplaybreaks \\
    \Sigma^y_i &= \hat{\mathcal{C}}\Sigma_0 \hat{\mathcal{C}}\t + \hat{\mathcal{C}} \Sigma^{\chi \zeta}_0 + (\Sigma^{\chi \zeta}_0)\t\hat{\mathcal{C}}\t + \hat{\Sigma}^\zeta,\allowdisplaybreaks \\
    \Sigma^{\chi\zeta}_0 &= \hat{\Sigma}^{\chi\zeta}_i,
\end{align}
\end{subequations}
where the constraints \eqref{eq:rouche_capelli_second_block_chi_covar} and \eqref{eq:rouche_capelli_second_block_chi_xi_covar} come from the second block of \eqref{eq:rouche_capelli} after multiplying by $\Sigma_k$ and $\Sigma_k^{\chi \zeta}$ from the right, respectively. Moreover, the constraint \eqref{eq:joint_covar_constraint} implies that the covariance matrix of the vector $a_k \triangleq [\chi_k\ G_k \chi_k\ \hat{\mathcal{D}}\zeta_k]\t$ must be positive semi-definite.

\begin{remark}\label{rem:non_exact_relax}
In contrast to \cite{rapakoulias_discrete-time_2023}, the relaxation $Y_k \succeq U_k \Sigma_k^{-1} U_k\t$ does not result in a lossless convexification for the present problem. This implies that, if the estimated terms in \eqref{eq:convex_dd_covar_problem_chi_state} were exactly equal to the true ones, then the relaxed terminal constraint $\Sigma^y_N \preceq \Sigma^y_f$ would be met. However, under the presence of estimation errors, 
satisfaction of this terminal condition is not guaranteed.
\end{remark}

\section{Numerical Examples}

This section evaluates the proposed framework through numerical simulations.
We consider the dynamical system 
\begin{equation*}
\begin{aligned}
    &A = 0.55\begin{bmatrix}
        -0.5 & 1.4 & 0.4\\
        -0.9 & 0.3 & -1.5\\
        1.1 & 1.0 & -0.4\\
    \end{bmatrix},\quad B=\begin{bmatrix}
        0.1 & -0.3\\
        -0.1 & -0.7\\
        0.7 & -1\\
    \end{bmatrix},\\
\end{aligned}
\end{equation*}
for which $y=[x_1, x_2]\t$, $\ell=2$, and $p\ell>n$.
The true matrices $\mathcal{A},\mathcal{B}$ were computed via the WFL on noise-free data.

\subsection{Model Estimation}
We quantify model estimation performance using the error metric $e = \left \Vert \begin{bmatrix} \mathcal{B} & \mathcal{A}\end{bmatrix} - \begin{bmatrix}\hat{\mathcal{B}} & \hat{\mathcal{A}}\end{bmatrix} \right\Vert_{\mathrm{F}}$.
The colored noise processes $w_{k+1} = 0.5 I w_k + \alpha_k$ and $q_{k+1} = 0.4 I q_k + \beta_k$ were used, where $\alpha_k \sim \mathcal{N}(0, I_3)$ and $\beta_k \sim \mathcal{N}(0, I_2)$.
The mean and variance of $e$ were computed for sample sizes ranging from 50 up to 5,000 samples, using 200 datasets of each size, and they are presented in Figure~\ref{fig:model_estimation_error}.
In addition to the TLS and IV methods, the ordinary least squares method (LS) was added for comparison.
For small sample sizes, the IV method exhibits comparable performance with LS and TLS, yet as the sample size increases, the performance gap quickly becomes evident.
Moreover, we observe that the performance gap between TLS and LS is relatively small, which is attributed to the fact that the truncated SVD approach used to compute the transformation matrix $L$ already provides some robustness against noise in the data.

\begin{figure}
    \centering
    \includegraphics[width=0.70\linewidth]{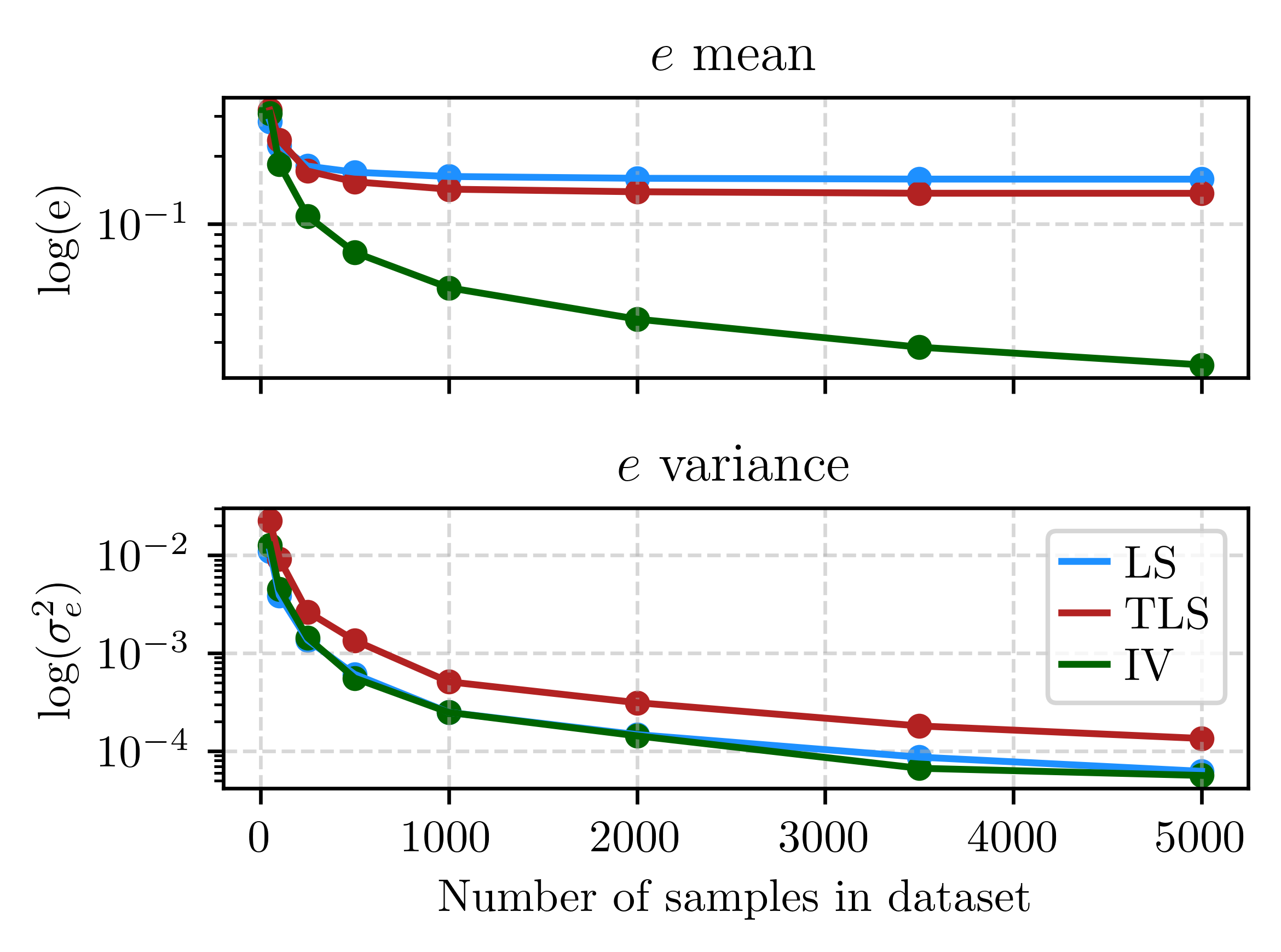}
    \caption{Model estimation error.}
    \label{fig:model_estimation_error}
\end{figure}

\subsection{Density Steering}\label{sec:density_steering}

Next, the algorithm is tested on a density steering example.
The initial output mean and covariance were $\mu_i^y = [0,\ 0]\t$ and $\Sigma_i^y = 2.5I_2$. 
A Lissajous reference trajectory was used with $\mu^y_f=[-5.9, -9.5]\t$ and $\Sigma^y_f=0.1 \Sigma_i^y$. 
The horizon length was $N=15$ steps, and the IV method was used for model and noise estimation.
The colored noise processes $w_{k+1}=0.1Iw_k + \alpha_k$ and $q_{k+1} = 0.1I q_k + \beta_k$ were used, where $\alpha_k \sim \mathcal{N}(0, 0.1^2I_3)$ and $\beta_k \sim \mathcal{N}(0, 0.1^2 I_2)$.
The distribution steering framework was run for various dataset sizes, using 500 datasets of each size.
For each dataset, the terminal output mean and covariance were computed empirically by simulating the controller on 2,000 Monte-Carlo trajectories. 
As a baseline for comparison, we consider the case of known noise in the collected data.

The mean steering performance is quantified by the terminal mean error norm $\mu_e=\Vert 
\mu^y_f - \mu^y_N \Vert$ and the covariance steering performance by the maximum eigenvalue of the difference between the desired and the achieved terminal covariance $\Sigma_{e} = \Sigma_f^y-\Sigma_N^y$.
The mean and covariance steering error metrics are presented in Figure~\ref{fig:ms_cs_error_metrics} and the Monte-Carlo trajectories are presented in Figure~\ref{fig:dd_direct_phase_plane}, where the blue line represents the optimal mean trajectory.
One notices an improvement in the distribution steering performance as the dataset size increases.
Finally, a behavior consistent with Remark~\ref{rem:non_exact_relax} is observed.
In particular, $\lambda_{\max}(\Sigma_e)$ is positive for both exact and estimated noise, but smaller when the noise is estimated.
As the dataset size increases, the mean error decreases and part of the robustness margin lost to noise estimation is recovered.

\begin{figure}
    \centering
    \includegraphics[width=0.70\linewidth]{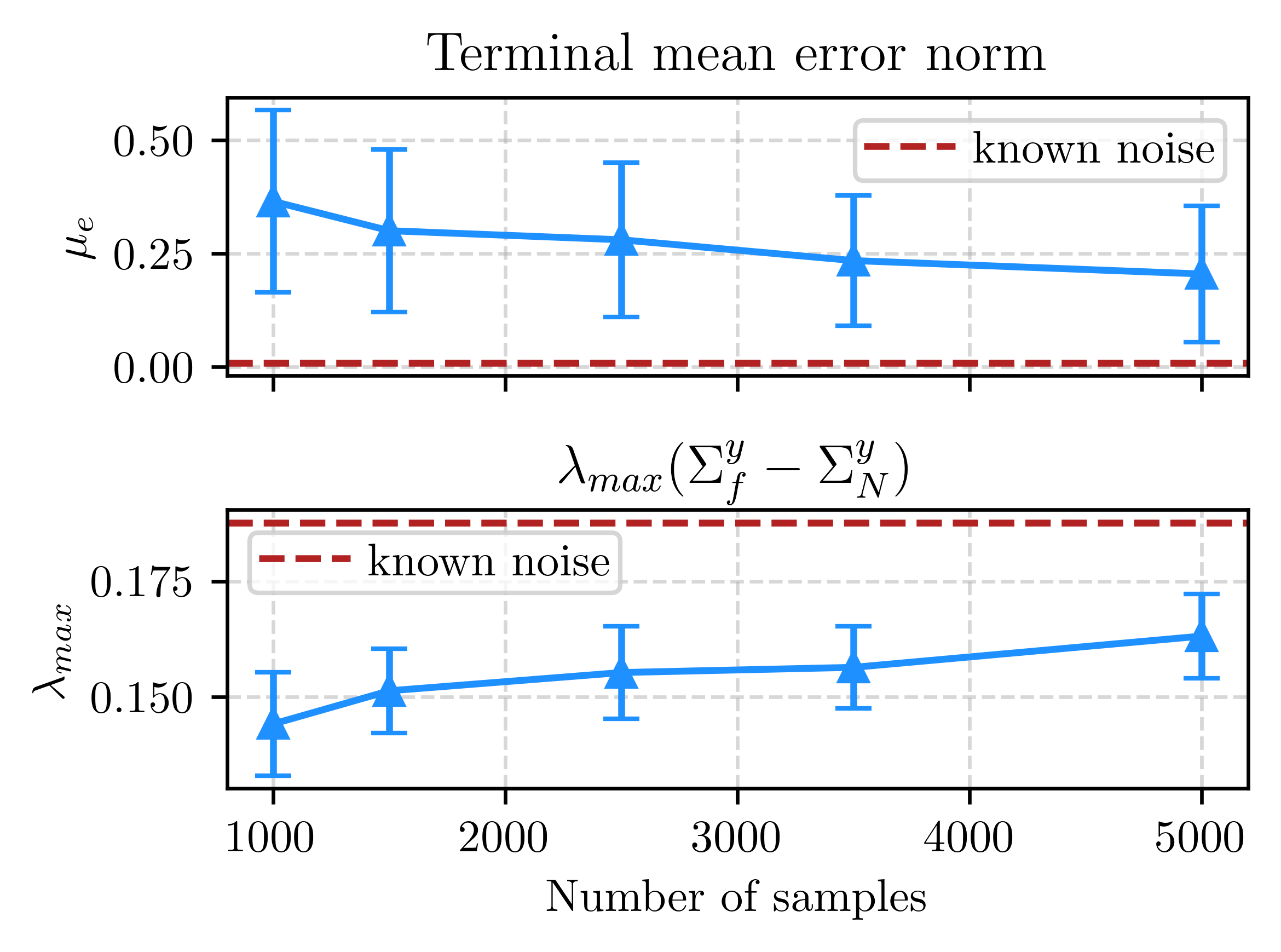}
    \caption{Terminal mean and covariance error.}
    \label{fig:ms_cs_error_metrics}
\end{figure}

\begin{figure}
    \centering
    \includegraphics[width=0.75\linewidth]{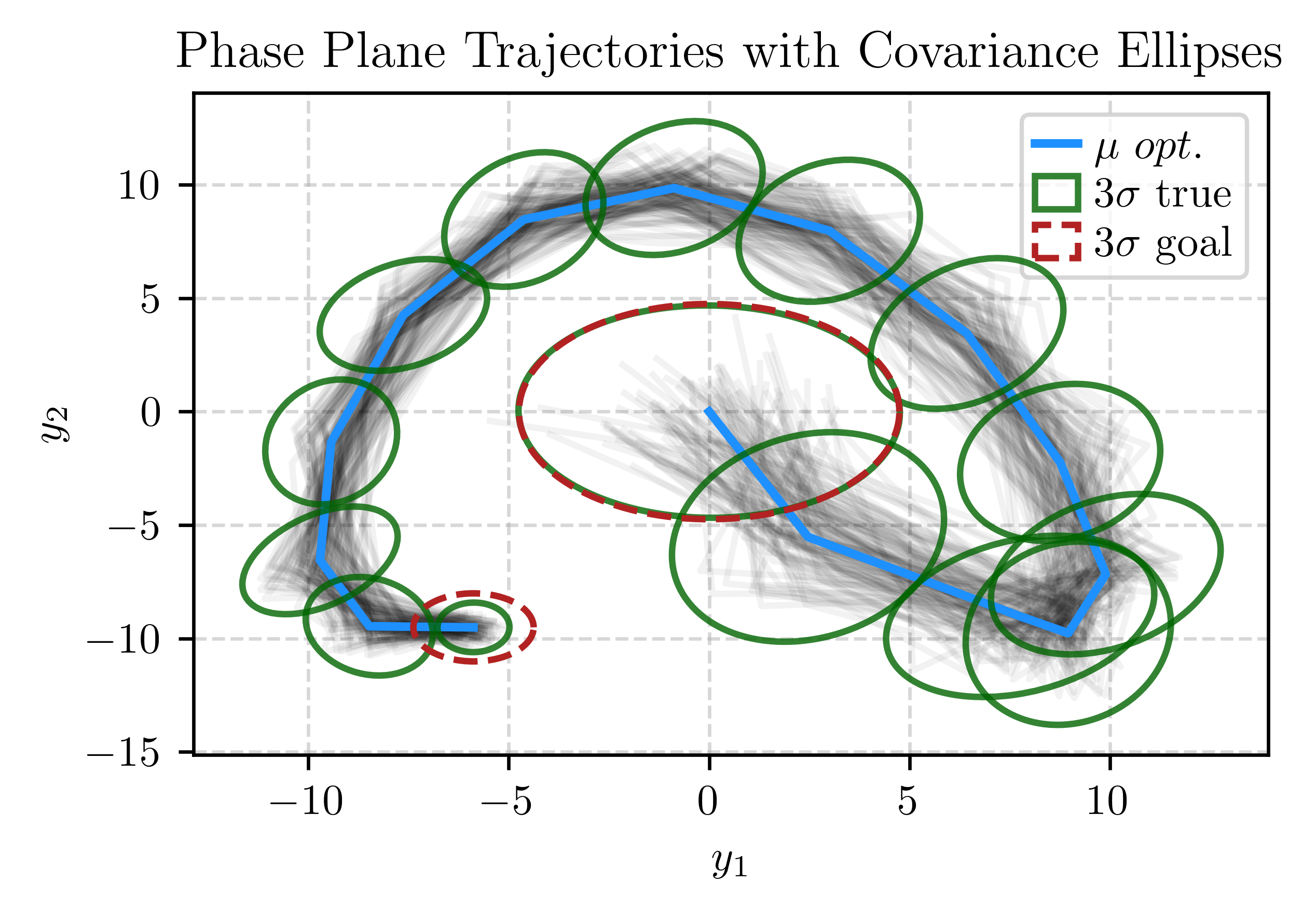}
    \caption{Output trajectories in phase plane.}
    \label{fig:dd_direct_phase_plane}
\end{figure}

\subsection{Comparison with a Model-Based Controller}
In this section, we compare our data driven framework with a model-based, output-feedback CS algorithm based on a Kalman filter \cite{pilipovsky_computationally_2024}.
In contrast to the proposed framework, the model based algorithm \cite{pilipovsky_computationally_2024} does not support colored process and measurement noise.
Therefore, we perform the comparison for i.i.d. noise processes $w,q$ in \eqref{eq:x_dynamics}, with $w_k \sim \mathcal{N}(0, 0.01 I)$ and $q_k \sim \mathcal{N}(0, 0.01I)$.
The rest of the experiment conditions are the same with Section~\ref{sec:density_steering}.

Figure~\ref{fig:model_based_comparison} presents the Monte-Carlo trajectories for the model-based controller.
As expected, the model-based controller achieves zero terminal mean error and it is not conservative with respect to covariance steering.

\begin{figure}
    \centering
    \includegraphics[width=0.75\linewidth]{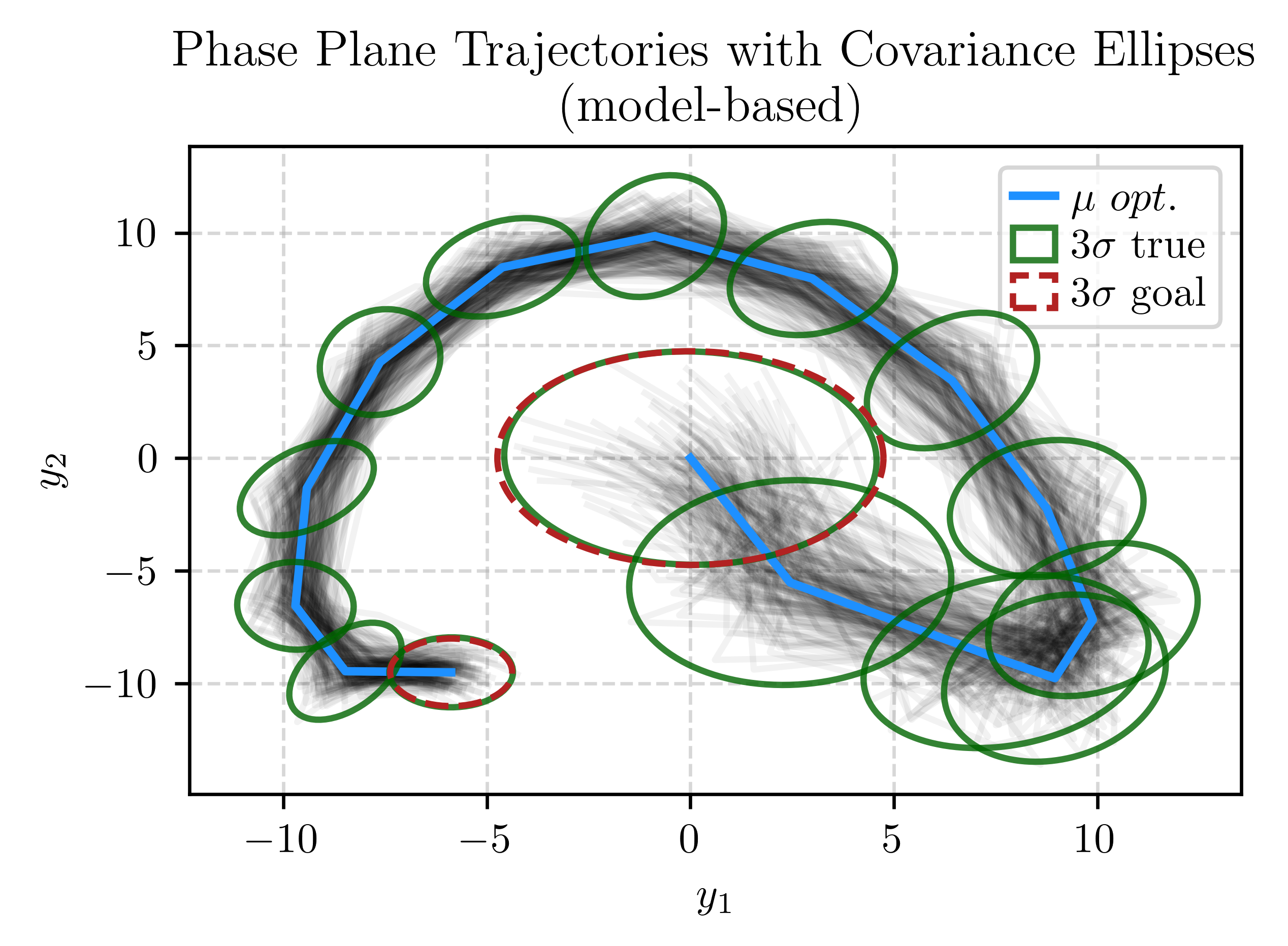}
    \caption{Output trajectories in phase plane (model-based).}
    \label{fig:model_based_comparison}
\end{figure}

\subsection{Computational Scalability Analysis}

Regarding the dimension $r$ of $\chi$, we can obtain an upper bound, independent of the observability index, as $r \le m(n+1)$.
To derive this bound, we exploit the relations $\ell \le n$ and $r=m\ell + n$.
For increasing $r$ and $N$, Table~\ref{tab:runtime_scaling} presents the number of decision variables and average solution times over five runs performed on a MacBook Air M4 with a maximum frequency of $4.4$~GHz and $16$~GB of RAM.

\begin{table}[ht]
\centering
\caption{Computational scalability analysis.}
\label{tab:runtime_scaling}
\small
\begin{tabular}{ccc|ccc}
\toprule
\multicolumn{3}{c|}{Steering horizon $N$ ($r=12$)} &
\multicolumn{3}{c}{Dimension $r$ ($N=15$)} \\
\cmidrule(r){1-3}\cmidrule(l){4-6}
$N$ & Vars & Time (s) &
$r$ & Vars & Time (s) \\
\midrule
5  & 2,800  & 0.832 & 4  & 1,440  & 0.840 \\
10 & 5,600  & 2.489 & 8  & 4,200  & 2.398 \\
20 &11,200  & 8.501 & 20 &21,120  &14.592 \\
30 &16,800  &17.684 & 28 &39,600  &33.714 \\
\bottomrule
\end{tabular}
\end{table}

\section{Conclusions and Future Work}

This paper presents a data-driven covariance steering framework for discrete-time linear systems with noisy, partial state measurements and temporally correlated process and measurement noise.
Several limitations motivate future work, including robust formulations that account for estimation uncertainty, as in \cite{pilipovsky_data_driven_2024}; extensions to non-stationary noise through online data collection and recursive system identification; and distributionally robust formulations for non-Gaussian noise.
Extending these principles to nonlinear systems is another important direction, potentially implementable through nonlinear system identification algorithms combined with iterative covariance steering.

\section{Acknowledgments}

The authors would like to thank Dr. Joshua Pilipovsky and George Rapakoulias for their constructive discussions.
Financial support for this work has been provided by  ONR award N00014-23-12353 and NSF award CNS-2219755.


\bibliographystyle{ieeetr}

\bibliography{references}


\section{APPENDIX}

\subsection{Non-minimal state dynamics matrices}

We denote by $\mathrm{Toep}(a,b)$ a Toeplitz matrix where
$a$ and $b$ are either vectors or stacked matrices of suitable dimensions.
We define the matrices
\begin{align*}
    \mathcal{C}_x &\triangleq[
        B,\ AB,\ \cdots,\ A^{\ell-1}B],\ \
    \mathcal{C}_w \triangleq [
        I,\ A,\ \cdots, \ A^{\ell-1}], \\
    \mathcal{T} &= \mathrm{Toep}(0, [0,\ CB, \ CAB, \ \cdots, \ CA^{\ell-2}B])\\
    \mathcal{T}_w &= \mathrm{Toep}(0, [0, \ C, \ CA, \ \cdots, \ CA^{\ell-2}])\\
    \mathcal{S} &\triangleq \begin{bmatrix}C(\mathcal{C}_x-A^{\ell}\mathcal{O}^\dagger\mathcal{T}) & CA^{\ell}\mathcal{O}^\dagger\end{bmatrix}.
\end{align*}

Then, the dynamics matrices $A_z, B_z, D_z$ in \eqref{eq:z_dynamics} are given in \eqref{eq:matrices}, where $\mathcal{S}_i$ ($i=1,\ldots, 2\ell$) is the $i$-th column of 
$\mathcal{S}$~\cite{zhao_direct_2025}. 

\begin{figure*}[!t]
\begin{equation} \label{eq:matrices}
    A_z = \begin{bmatrix}
        0 & 0 & \cdots & 0 & 0 & 0 & 0 & \cdots & 0 & 0 \\
        I_m & 0 & \cdots & 0 & 0 & 0 & 0 & \cdots & 0 & 0 \\
        0 & I_m & \cdots & 0 & 0 & 0 & 0 & \cdots & 0 & 0 \\
        \vdots & \vdots & \ddots & \vdots & \vdots & \vdots & \vdots & \ddots & \vdots & \vdots \\
        0 & 0 & \cdots & I_m & 0 & 0 & 0 & \cdots & 0 & 0 \\
        \hline \\
        \mathcal{S}_1 & \mathcal{S}_2 & \cdots & \mathcal{S}_{\ell-1} & \mathcal{S}_{\ell} & \mathcal{S}_{\ell+1} & \mathcal{S}_{\ell+2} & \cdots & \mathcal{S}_{2\ell-1} & \mathcal{S}_{2\ell} \\
        0 & 0 & \cdots & 0 & 0 & I_p & 0 & \cdots & 0 & 0 \\
        0 & 0 & \cdots & 0 & 0 & 0 & I_p & \cdots & 0 & 0 \\
        \vdots & \vdots & \ddots & \vdots & \vdots & \vdots & \vdots & \ddots & \vdots & \vdots \\
        0 & 0 & \cdots & 0 & 0 & 0 & 0 & \cdots & I_p & 0 \\
    \end{bmatrix},
\quad
B_z = 
    \begin{bmatrix}
        I_m \\ 0 \\ 0 \\ \vdots \\ 0 \\ \hline \\ 0 \\ 0 \\ 0 \\ \vdots \\ 0
    \end{bmatrix},
\quad
D_z = 
     \begin{bmatrix}
        0 \\ 0 \\ 0 \\ \vdots \\ 0 \\ \hline \\ I_p \\ 0 \\ 0 \\ \vdots \\ 0
    \end{bmatrix}.
\end{equation}
\end{figure*}

\subsection{Proof of Lemma~\ref{lem:d_noise_correlation}}\label{sec:appendix_colored_noise_proof}
    The zero-mean, Gaussian nature of $\tilde{\zeta}_k$ is evident from the fact that $\tilde{\zeta}_k$ is a linear combination of zero-mean Gaussian random variables, namely, $\tilde{w}_k,\tilde{q}_k$, and $q_k$.
    To show that the samples of the noise process $\tilde{\zeta}_k$ are correlated in time, we compute $\mathbb{E}[\tilde{\zeta}_k \tilde{\zeta}_m\t]$ assuming, without loss of generality, that $m > k$. The expectation  $\mathbb{E}[\tilde{\zeta}_k \tilde{\zeta}_m\t]$ is expanded as
\begin{align*}
    \mathbb{E}[\tilde{\zeta}_k \tilde{\zeta}_m\t] 
    &= \mathbb{E}[F_w \tilde{w}_k \tilde{w}_m\t F_w\t 
    + F_w \tilde{w}_k \tilde{q}_m\t F_q\t 
    + F_w \tilde{w}_k q_m\t \\
    &+ F_q \tilde{q}_k \tilde{w}_m\t F_w\t 
    + F_q \tilde{q}_k \tilde{q}_m\t F_q\t 
    + F_q \tilde{q}_k q_m\t \\
    &+ q_k \tilde{w}_m\t F_w\t 
    + q_k \tilde{q}_m\t F_q\t 
    + q_k q_m\t].
\end{align*}
Notice that since $w_k$ is independent of $q_m$ for all times, it follows that $\mathbb{E}[\tilde{w}_k\tilde{q}_m\t + \tilde{w}_kq_m\t + \tilde{q}_k \tilde{w}_m\t + q_k\tilde{w}_m\t] = 0$.
Moreover, since we assumed $m>k$, it holds that $\mathbb{E}[\tilde{q}_k q_m\t] = 0$.
Therefore, the expression for $\mathbb{E}[\tilde{\zeta}_k \tilde{\zeta}_m\t]$ becomes
\begin{align*}
     \mathbb{E}[\tilde{\zeta}_k \tilde{\zeta}_m\t] &= \mathbb{E}[F_w\tilde{w}_k\tilde{w}_m\t F_w\t + F_q\tilde{q}_k\tilde{q}_m\t F_q\t  + q_k\tilde{q}_m\t F_q\t + q_kq_m\t].
\end{align*}

For $m\ne k$, and under the i.i.d. assumption for $q_k$, it holds that $\mathbb{E}[q_kq_m\t] = 0$ and, similarly, $\mathbb{E}[w_kw_m\t]=0$. Recalling that $\tilde{w}_k,\tilde{q}_k$ contain the elements of $w$ and $q$ from time $k-1$ to time $k-\ell$, one sees that for $k-1 \ge m-\ell$, the matrices $\tilde{w}_k\tilde{w}_m\t$ and $\tilde{q}_k\tilde{q}_m\t$ contain elements of the form $w_iw_i\t$ and $q_iq_i\t$, respectively, whose expected value might not be zero.
Moreover, for $m-k \le \ell$, the vector $\tilde{q}_m$ includes $q_k$, and thus, the matrix $q_k\tilde{q}_m\t$ contains elements of the form $q_kq_k\t$, which can have non-zero expected value.

Therefore, for $|m-k| \le \ell$, there exist $k,m$ with $k\ne m,$ such that $\mathbb{E}[\tilde{\zeta}_k \tilde{\zeta}_m\t] \ne 0$.
For $|m-k| > \ell$, the vectors $\tilde{\zeta}_k$ and $\tilde{\zeta}_m$ are independent.\qed

\subsection{Covariance Propagation Equations}
The covariance $\Sigma_{k+1} = \mathbb{E}[(\chi_{k+1} - \mu_{k+1})(\chi_{k+1}-\mu_{k+1})\t]$ can be expanded as follows
\begin{equation*}
\begin{aligned}
        &\Sigma_{k+1} = \mathbb{E}[(\chi_{k+1}-\mu_{k+1})(\chi_{k+1}-\mu_{k+1})\t] \\
        &=\mathbb{E}[(\mathcal{A}\chi_k+ \mathcal{B}u_k + \xi_k)(\mathcal{A}\chi_k+ \mathcal{B}u_k + \xi_k)\t] -\mu_{k+1}\mu_{k+1}\t \\
        &=\mathbb{E}[\mathcal{A}\chi_k\chi_k\t\mathcal{A}\t + \mathcal{A}\chi_k(\chi_k-\mu_k)\t K_k\t\mathcal{B}\t + \mathcal{A}\chi_kv_k\t\mathcal{B}\t\\
        &+\mathcal{A}\chi_k\nu_k\t B\t + \mathcal{A}\chi_k\xi_k\t + \mathcal{B}K(\chi_k-\mu_k)\chi_k\t\mathcal{A}\t\\
        &+\mathcal{B}K(\chi_k-\mu_k)(\chi_k-\mu_k)\t K\t\mathcal{B}\t + \mathcal{B}K(\chi_k-\mu_k)v_k\t\mathcal{B}\t\\
        &+\mathcal{B}K(\chi_k - \mu_k)\nu_k\t\mathcal{B}\t + \mathcal{B}K(\chi_k - \mu_k)\xi_k\t + \mathcal{B}v_k\chi_k\t\mathcal{A}\t\\
        &+\mathcal{B}v_k(\chi_k-\mu_k)\t K\t\mathcal{B}\t + \mathcal{B}v_kv_k\t\mathcal{B}\t + \mathcal{B}v_k\nu_k\t\mathcal{B}\t\\
        &+\mathcal{B}v_k\xi_k\t + \mathcal{B}\nu_k\chi_k\t\mathcal{A}\t + \mathcal{B}\nu_k(\chi_k-\mu_k)\t K\t\mathcal{B}\t\\
        &+\mathcal{B}\nu_kv_k\t\mathcal{B}\t + \mathcal{B}\nu_k\nu_k\t\mathcal{B}\t + \mathcal{B}\nu_k\xi_k\t + \xi_k\chi_k\t\mathcal{A}\t + \xi_kv_k\t\mathcal{B}\t\\
        &+\xi_k(\chi_k-\mu_k)\t K\t\mathcal{B}\t + \xi_k\nu_k\t\mathcal{B}\t + \xi_k\xi_k\t\\
        &-\mathcal{A}\mu_k\mu_k\t\mathcal{A}\t -\mathcal{A}\mu_kv_k\t\mathcal{B}\t -\mathcal{B}v_k\mu_k\t\mathcal{A}\t -\mathcal{B}v_kv_k\t\mathcal{B}\t].
\end{aligned}
\end{equation*}
By nulling the expectations of products of independent variables one of which is zero-mean, and by noting that $\mathbb{E}[x_k\mu_k\t] = \mu_k\mu_k\t$, the covariance propagation equation becomes
\begin{equation*}
    \begin{aligned}
        \Sigma_{k+1} &= \mathcal{A}\Sigma_k\mathcal{A}\t + \mathcal{A}\Sigma_kK_k\t\mathcal{B}\t + \mathcal{A}\Sigma^{\chi\xi}_k + \mathcal{B}K_k\Sigma_k\mathcal{A}\t \\
        &+ \mathcal{B}K_k\Sigma_kK_k\t\mathcal{B}\t + \mathcal{B}K_k\Sigma^{\chi\xi}_k +(\Sigma^{\chi\xi}_k)\t\mathcal{A}\t\\
        &+(\Sigma^{\chi\xi}_k)\t K_k\t\mathcal{B}\t + \Sigma^\xi,
    \end{aligned}
\end{equation*}
where $\Sigma^\xi_k \triangleq \mathbb{E}[\xi_k\xi_k\t] \equiv \Sigma^\xi$ and $\Sigma^{\chi \xi}_k \triangleq \mathbb{E}[\chi_k\xi_k\t] = \mathbb{E}[\chi_k \zeta_k\t] \mathcal{D}\t$. 
Since the noise is zero-mean, the second non-central cross-moment between $\chi_k$ and $\xi_k$, as well as the second moment of $\xi_k$, are the same as the central ones. Defining $\Sigma^{\chi \zeta}_k \triangleq \mathbb{E}[\chi_k \zeta_k\t]$ and $\Sigma^\zeta_k \triangleq \mathbb{E}[\zeta_k \zeta_k\t] \equiv \Sigma^\zeta$, and substituting $\Sigma^{\chi \xi}_k = \Sigma^{\chi \zeta}_k \mathcal{D}\t$ and $\Sigma^\xi = \mathcal{D} \Sigma^\zeta \mathcal{D}\t$ we get \eqref{eq:chi_covar_propagation}.

Considering the noise model \eqref{eq:zeta_noise_dynamics}, the cross-covariance propagation equation can be derived as follows
\begin{equation*}
    \begin{aligned}
        \Sigma^{\chi\zeta}_{k+1} &= \mathbb{E}[\chi_{k+1}\zeta_{k+1}\t]  \\
        &=\mathbb{E}[((\mathcal{A}\chi_k + \mathcal{B} u_k + \mathcal{D}\zeta_k)(\Psi\zeta_k + \eta_k)\t] \\
        &=\mathbb{E}[\mathcal{A}\chi_k\zeta_k\t\Psi\t + \mathcal{A}\chi_k\eta_k\t + \mathcal{B}v_k\zeta_k\t\Psi\t + \mathcal{B}v_k\eta_k\t\\
        &+\mathcal{B}\nu_k\zeta_k\t\Psi\t + \mathcal{B}\nu_k\eta_k\t + \mathcal{B}K\chi_k\zeta_k\t\Psi\t + \mathcal{B}K\chi_k\eta_k\t\\
        &-\mathcal{B}K\mu_k\zeta_k\t\Psi\t -\mathcal{B}K\mu_k\eta_k\t + \mathcal{D}\zeta_k\zeta_k\t \Psi\t + \xi_k\eta_k\t] \\
        &=(\mathcal{A} + \mathcal{B}K_k)\Sigma^{\chi\zeta}_k\Psi\t + \mathcal{D}\Sigma^\zeta \Psi\t.
    \end{aligned}
\end{equation*}

Finally, the output covariance $\Sigma^y_k$ is computed as follows
\begin{equation*}
    \begin{aligned}
        \Sigma^y_k &= \mathbb{E}[(y_k-\mu^y_k)(y_k-\mu^y_k)\t]\\
     &= \mathbb{E}[(\mathcal{C}\chi_k +\zeta_k)(\mathcal{C}\chi_k+\zeta_k)\t] - \mu^y_k(\mu^y_k)\t \\
      &= \mathbb{E}[\mathcal{C}\chi_k \chi_k\t\mathcal{C}\t + \mathcal{C}\chi_k \zeta_k\t + \zeta_k \chi_k\t \mathcal{C}\t + \zeta_k\zeta_k\t] - \mathcal{C}\mu_k \mu_k\t \mathcal{C}\t \\
   &= \mathcal{C}\Sigma_k \mathcal{C}\t + \mathcal{C} \Sigma^{\chi \zeta}_k + (\Sigma^{\chi \zeta}_k)\t\mathcal{C}\t + \Sigma_\zeta.
    \end{aligned}
\end{equation*}


\end{document}